\pdfoutput=1
\newif\ifFull
\Fullfalse

\documentclass{llncs}

\usepackage{amsmath}
\usepackage{amssymb}
\usepackage{algorithmic}
\usepackage[boxed]{algorithm}
\usepackage{tikz}
\usetikzlibrary{shapes.multipart}
\begin{document}
\ifFull\else
\pagestyle{plain}
\fi

\title{External-Memory Network Analysis Algorithms for Naturally Sparse Graphs}

\author{Michael T. Goodrich \and Pawe\l{} Pszona}

\institute{
Dept.~of Computer Science \\
University of California, Irvine
}

\maketitle

\begin{abstract}
In this paper,
we present a number of network-analysis algorithms in the
external-memory model.
We focus on methods for
large naturally sparse graphs, that is, $n$-vertex graphs that
have $O(n)$ edges and are structured so that this sparsity property
holds for any subgraph of such a graph.
We give efficient external-memory algorithms for the following
problems for such graphs:
\begin{enumerate}
\setlength{\itemsep}{0pt}
\item
Finding an approximate $d$-degeneracy ordering.
\item
Finding a cycle of length exactly $c$.
\item
Enumerating all maximal cliques.
\end{enumerate}
Such problems are of interest, for example, in the analysis of social
networks, where they are used to study network cohesion.
\end{abstract}

\section{Introduction}
Network analysis studies the structure of relationships between
various entities, with those entities represented as vertices
in a graph and their relationships represented as edges in that
graph (e.g., see~\cite{Doreian1994267}).
For example, such structural analyses
include link-analysis for Web graphs,
centrality and cohesion measures in social networks,
and network motifs in biological networks.
In this paper, we are particularly
interested in network analysis algorithms for finding various kinds
of small subgraphs and graph partitions
in large graphs that are likely to occur in practice.
Of course, this begs the question of what kinds of graphs are likely
to occur in practice.

\subsection{Naturally Sparse Graphs}
A network property addressing the concept of a
``real world'' graph that is gaining in prominence is
the \emph{$k$-core number}~\cite{Seidman1983269},
which is equivalent to a graph's \emph{width}~\cite{Freuder:1982},
\emph{linkage}~\cite{kirousis:626},
\emph{$k$-inductivity}~\cite{i-cigo-94},
and \emph{$k$-degeneracy}~\cite{aks-lids-87,lw-kdg-70},
and is one less than its
Erd{\H{o}}s-Hajnal coloring number~\cite{erdos66}.
A $k$-core, $G'$, in a graph, $G$, is a maximal connected subgraph
of $G$ such that each vertex in $G'$ has degree at least $k$.
The $k$-core number of a graph $G$ is the maximum $k$ such that $G$
has a non-empty $k$-core.
We say that a graph $G$ is \emph{naturally sparse} if its $k$-core
number is $O(1)$.
This terminology is motivated by the fact that
almost every $n$-vertex graph with $O(n)$ edges has a bounded
$k$-core number, since
Pittel {\it et al.}~\cite{Pittel1996111} show that
a random graph with $n$ vertices and $cn$ edges
(in the Erd{\H{o}}s-R{\'e}nyi model)
has $k$-core number at most $2c+o(c)$, with high probability.
\pagebreak
Riordan \cite{riordan-2008-17} and Fernholz and Ramachandran \cite{fernholz}
have also studied $k$-cores in random graphs.

In addition, we also have the following:
\begin{itemize}
\setlength{\itemsep}{0pt}
\item
Every $s$-vertex subgraph of a naturally sparse graph is naturally
sparse, hence, has $O(s)$ edges.
\item
Any planar graph has $k$-core number at most $5$, hence, is naturally
sparse.
\item
Any graph with bounded arboricity is naturally sparse
(e.g., see~\cite{chrobak91}).
\item
Eppstein and Strash~\cite{es-real-11} verify experimentally that real-world
graphs in four different data repositories all have small $k$-core
numbers relative to their sizes; hence, these real-world
graphs give an empirical motivation for naturally sparse graphs.
\item
Any network generated by
the Barab{\'a}si-Albert~\cite{Barabi1999} preferential attachment process,
with $m\in O(1)$, or as in Kleinberg's small-world
model~\cite{k-tswp-00}, is naturally sparse.
\end{itemize}
Of course, one can artificially define an $n$-vertex graph, $G'$, with $O(n)$
edges that is not naturally sparse just by creating a clique of
$O(n^{1/2})$ vertices in an $n$-vertex
graph, $G$, having $O(n)$ edges. We would argue,
however, that such a graph $G'$ would not arise ``naturally.''
We are interested in algorithms for large, naturally
sparse graphs.

\subsection{External-Memory Algorithms}
One well-recognized way of designing algorithms for processing
large data sets is to formulate such algorithms in the
\emph{external memory model} (e.g., see the
excellent survey by Vitter \cite{vitter}).
In this model,
we have a single
CPU with main memory capable of storing $M$ items and that computer
is connected to $D$ external disks that are capable of storing a
much larger amount of data.
Initially, we assume the parallel disks are storing an input of size $N$.
A single I/O between one of the external disks and main memory
is defined as either reading a block of $B$ consecutively stored items
into memory or writing a block of the same size to a disk.
Moreover, we assume that this can be done
on all $D$ disks in parallel if need be.

Two fundamental primitives of the model are \emph{scanning} and \emph{sorting}.
Scanning is the operation of streaming $N$ items stored on $D$ disks through
main memory, with I/O complexity
\begin{displaymath}
scan(N) = \Theta\left(\frac{N}{DB}\right),
\end{displaymath}
and sorting $N$ items has I/O complexity
\begin{displaymath}
sort(N) = \Theta\left(\frac{N}{DB}\log_{M/B}\frac{N}{B} \right),
\end{displaymath}
e.g., see Vitter~\cite{vitter}.

Since this paper concerns graphs, we assume
a problem instance is a graph $G=(V,E)$, with $n = |V|$, $m = |E|$ and
$N = |G| = m + n$.
If $G$ is $d$-degenerate, that is, has $k$-core number, $d$,
then $m \leq dn$ and $N = O(dn) = O(n)$ for $d=O(1)$.
We use $d$ to denote the $k$-core number of an input graph,
$G$, and we use the term ``$d$-degenerate'' as a shorthand for
``$k$-core number equal to $d$.''

\subsection{Previous Related Work}
Several researchers have studied algorithms for
graphs with bounded $k$-core numbers (e.g.,
see~\cite{AloGut-Algo-09,short_cycles,els2010,GolVil-WG-08,i-cigo-94}).
These methods are often based on the fact that the vertices
in a graph with $k$-core number, $d$, can be ordered by
repeatedly removing a vertex of degree at most $d$,
which gives rise to a numbering of the
vertices, called a \emph{$d$-degeneracy ordering}
or \emph{Erd{\H{o}}s-Hajnal sequence},
such that each vertex has
at most $d$ edges to higher-numbered vertices.
In the RAM model, this greedy algorithm
takes $O(n)$ time
(e.g., see~\cite{batagelj2003}).
Bauer {\it et al.}~\cite{bkw-eglkd-10} describe methods for generating
such graphs and their $d$-degeneracy orderings at random.

In the internal-memory RAM model, Eppstein {\it et al.}~\cite{els2010}
show how to find all maximal cliques in a $d$-degenerate graph in
$O(d3^{d/3}n)$ time.
Alon {\it et al.}~\cite{short_cycles} show that one can find a cycle
of length exactly $c$, or show that one does not exist, in
a $d$-degenerate graph in time
$O(d^{1-1/k}m^{2-1/k})$, if $c=4k-2$,
time $O(dm^{2-1/k})$, if $c=4k-1$ or $4k$,
and
time
$O(d^{1+1/k}m^{2-1/k})$, if $c=4k+1$.

A closely related concept to a $d$-degeneracy ordering is
a \emph{$k$-core decomposition} of a graph, which is
a labeling of each vertex $v$ with the largest $k$
such that $v$ belongs to a $k$-core.
Such a labeling can also be produced by the simple
linear-time greedy algorithm
that removes a vertex of minimum degree with each iteration.
Cheng {\it et al.}~\cite{ckco-ecdmn-11} describe recently an external-memory
method for constructing a $k$-core decomposition,  but
their method is unfortunately fatally flawed\footnote{We contacted the
  authors and they confirmed that their method is
  indeed incorrect.}.
The challenge in producing a $k$-core decomposition or $d$-degeneracy
ordering in external memory is that the standard greedy method, which
works so well in internal memory, can cause a large number of I/Os
when implemented in external memory.
Thus, new approaches are needed.

\subsection{Our Results}
In this paper, we present efficient external-memory
network analysis algorithms for
naturally sparse graphs (i.e., degenerate graphs with small degeneracy).
First, we give a simple algorithm
for computing a $(2+\epsilon)d$-degeneracy ordering
of a $d$-degenerate
graph $G=(V,E)$, without the need to know the value of $d$
in advance.
The I/O complexity of our algorithm is $O(sort(dn))$.

Second, we give an algorithm
for determining whether a $d$-degenerate
graph $G=(V,E)$ contains a simple cycle of a fixed length $c$.
This algorithm uses
$O\Big(d^{1\pm\epsilon}\cdot \big(k\cdot sort(m^{2-\frac{1}{k}}) + (4k)!\cdot scan(m^{2-\frac{1}{k}})\big)\Big)$
I/O complexity,
where $\epsilon$ is a constant depending on $c\in\{4k-2,\ldots,4k+1\}$.

Finally, we present an algorithm
for listing all maximal cliques of an undirected $d$-degenerate
graph $G=(V,E)$, with $O(3^{\delta/3}sort(dn))$ I/O complexity,
where $\delta = (2+\epsilon)d$.

One of the key insights to our second and third results is to show
that, for the sake of designing efficient external-memory algorithms,
using a $(2+\epsilon)d$-degeneracy ordering is almost as good as a
$d$-degeneracy ordering.
In addition to this insight, there are a number of technical details
that lead to our results, which we outline in the remainder of this
manuscript.

\section{Approximating a $d$-Degeneracy Ordering} \label{degeneracy_section}

Our method for constructing a $(2+\epsilon)d$-degeneracy ordering
for a $d$-degenerate graph, $G=(V,E)$,
is quite simple and is given below as
Algorithm~\ref{alg1}.
Note that our algorithm does not take into account the value of $d$,
but it assumes we are given a constant $\epsilon>0$ as part of the
input.
Also, note that this algorithm destroys $G$ in the process. If one
desires to maintain $G$ for other purposes, then one should first
create a backup copy of $G$.

\begin{algorithm}
\begin{algorithmic}[1]
\STATE $L \leftarrow \emptyset$
\WHILE{$G$ is nonempty}
  \STATE $S \leftarrow n\epsilon /(2+\epsilon)$ vertices of smallest degree in $G$
  \STATE $L \leftarrow L | S$ \emph{$\qquad\qquad\qquad$ // append $S$ to the end of $L$}
  \STATE remove $S$ from $G$
\ENDWHILE
\RETURN $L$
\end{algorithmic}
\caption{Approximate degeneracy ordering of vertices\vspace*{-18pt}}
\label{alg_degorder}
\label{alg1}
\end{algorithm}
\begin{lemma} \label{alg_degorder_lemma_1}
If $G$ is a $d$-degenerate graph, then
Algorithm \ref{alg_degorder} computes a $(2+\epsilon)d$-degeneracy
ordering of $G$.
\end{lemma}
\begin{proof}
Observe that any $d$-degenerate graph with $n$
vertices has at most $2n/c$ vertices of
degree at least $cd$.
Thus, $G$ has at most $2n/(2+\epsilon)$ vertices of degree
at least $(2+\epsilon)d$.
This means that the $n\epsilon/(2+\epsilon)$ vertices of smallest
degree in $G$ each have degree at most $(2+\epsilon)d$.
Therefore, every element of set $S$ created in line 3
has at most $(2+\epsilon)d$ neighbors in (the remaining graph) $G$.
When we add $S$ to $L$ in line 4,
we keep the property that every element of $L$ has at most
$(2+\epsilon)d$ neighbors in $G$ that are placed behind it in $L$.
Furthermore,
note that,
after we remove vertices in $S$
(and their incident edges) from $G$ in line 5, $G$ is still at
most $d$-degenerate (every subgraph of a
$d$-degenerate graph is at most $d$-degenerate);
hence, an inductive argument applies to the remainder of the
algorithm.
\qed
\end{proof}

Note that,
after $\lceil\log_{(2+\epsilon)/2} (dn)\rceil = O(\lg n)$ iterations,
we must have processed all of $G$ and placed all its vertices
on $L$, which is a $(2+\epsilon)d$-degeneracy ordering
for $G$ and that this property holds even though
the algorithm does not take the value of $d$ into account.

The following lemma is proved in the Appendix.

\begin{lemma} \label{alg_degorder_lemma_2}
An iteration of
the \verb\while\ loop (lines 3-5)
of Algorithm \ref{alg_degorder} can be
implemented in $O(sort(dn))$ I/O's in the external-memory model,
where $n$ is the number of vertices
in $G$ at the beginning of the iteration.
\end{lemma}

Thus, we have the following.
\begin{theorem}
We can compute a $(2+\epsilon)d$-degeneracy ordering
of a $d$-degenerate graph, $G$,
in $O(sort(dn))$ I/O's in the external-memory model, without knowing
the value of $d$ in advance.
\end{theorem}
\begin{proof}
Since the number of vertices of $G$ decreases by a factor of $2/(2+\epsilon)$
in each iteration, and each iteration uses
$O(sort(dn))$ I/O's, where $n$ is the number of vertices in $G$
at the beginning of the iteration
(by Lemma \ref{alg_degorder_lemma_2}), the total number of I/O's, $I(G)$, is
bounded by
\begin{eqnarray*}
I(G) & = & O\left(sort(dn) + sort\Big(\big(2/(2+\epsilon)\big)dn\Big) +
  sort\Big(\big(2/(2+\epsilon)\big)^2dn \Big) + \cdots\right) \\
& = & O\left(sort(dn)\bigg(1 + \frac{2}{2+\epsilon} +
  \Big(\frac{2}{2+\epsilon}\Big)^2 + \cdots\bigg) \right) \\
& = & O(sort(dn)). \\[-36pt]
\end{eqnarray*}
\qed
\end{proof}

This theorem hints at the possibility of
effectively using a $(2+\epsilon)d$-degeneracy ordering in place of
a $d$-degeneracy ordering in external-memory algorithms for naturally
sparse graphs.
As we show in the remainder of this paper,
achieving this goal is indeed possible, albeit with some
additional alterations from
previous internal-memory algorithms.

\section{Short Paths and Cycles} \label{cycles_section}

In this section,
we present external-memory algorithms for finding short cycles
in directed or undirected graphs.
Our approach is an external-memory adaptation
of internal-memory algorithms
by Alon {\it et al.}~\cite{short_cycles}.
We begin with the definition and an example of a \emph{representative}
due to Monien \cite{monien_paths}. A $p$-set is a set of size $p$.

\begin{definition}[representative]
Let $\mathcal{F}$ be a collection of $p$-sets. A sub-collection $\widehat{\mathcal{F}} \subseteq\mathcal{F}$
is a \emph{$q$-representative} for $\mathcal{F}$, if for every $q$-set $B$, there exists a set
$A \in \mathcal{F}$ such that $A \cap B = \emptyset$ if and only if there exists a set
$\widehat{A} \in \widehat{\mathcal{F}}$ with this property.
\end{definition}

Every collection of $p$-sets $\mathcal{F}$ has a
$q$-representative $\widehat{\mathcal{F}}$ of size at most $\binom{p+q}{p}$
(from Bollob{\'a}s~\cite{springerlink:10.1007/BF01904851}).
An optimal representative, however, seems difficult to find.
Monien \cite{monien_paths} gives a construction of representatives of size
at most $O(\sum_{i=1}^{q}p^i)$.
It uses a $p$-ary tree of height $\leq q$ with the following properties.
\begin{itemize}
\item Each node is labeled with either a set $A \in \mathcal{F}$ or a special symbol $\lambda$.

\item If a node is labeled with a set $A$ and its depth is less than $q$, it has
exactly $p$ children, edges to which are labeled with elements from $A$ (one element per edge,
every element of $A$ is used to label exactly one edge).

\item If a node is labeled with $\lambda$ or has depth $q$, it has no children.

\item Let $E(v)$ denote the set of all edge labels on the way from the vertex $v$ to the root of the tree.
Then, for every $v$:
  \begin{itemize}
  \item[--] if $v$ is labeled with $A$, then $A \cap E(v) = \emptyset$

  \item[--] if $v$ is labeled with $\lambda$, then there are no $A \in \mathcal{F}$ s.t.
   $A \cap E(v) = \emptyset$.
  \end{itemize}
\end{itemize}

Monien shows that if a tree $T$ fulfills the above conditions,
defining $\widehat{\mathcal{F}}$ to be the set of all labels of the tree's nodes yields
a $q$-representative for $\mathcal{F}$.
As an example, consider a collection of $2$-sets, $\mathcal{F} = \{\{2,4\}, \{1,5\}, \{1,6\},
\{1,7\}, \{3,6\}, \{3,8\}, \\ \{4,7\}, \{4,8\}\}$.
Fig. \ref{fig_3} presents $\widehat{\mathcal{F}}$, a $3$-representative of $\mathcal{F}$ in the tree form.

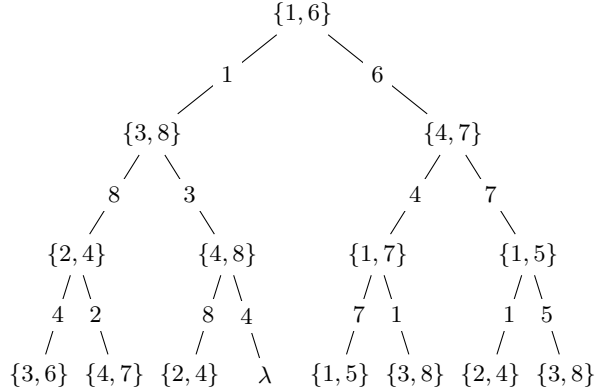
\begin{figure}
\begin{center}
\begin{tikzpicture}
[level distance=16mm,
 every node/.style={font=\footnotesize},
 reducer/.style={inner sep=1pt, font=\footnotesize, rectangle split, rectangle split parts=2,
                 rectangle split empty part height=2ex},
 level 1/.style={sibling distance=40mm},
 level 2/.style={sibling distance=20mm},
 level 3/.style={sibling distance=10mm},
]

\node (root) {$\{1,6\}$}
  child { node {$\{3,8\}$}
    child { node {$\{2,4\}$}
      child { node {$\{3,6\}$}
        edge from parent node[fill=white] {$4$}
      }
      child { node {$\{4,7\}$}
        edge from parent node[fill=white] {$2$}
      }
      edge from parent node[fill=white] {$8$}
    }
    child { node {$\{4,8\}$}
      child { node {$\{2,4\}$}
        edge from parent node[fill=white] {$8$}
      }
      child { node {$\lambda$}
        edge from parent node[fill=white] {$4$}
      }
      edge from parent node[fill=white] {$3$}
    }
    edge from parent node[fill=white] {$1$}
  }
  child { node {$\{4,7\}$}
    child { node {$\{1,7\}$}
      child { node {$\{1,5\}$}
        edge from parent node[fill=white] {$7$}
      }
      child { node {$\{3,8\}$}
        edge from parent node[fill=white] {$1$}
      }
      edge from parent node[fill=white] {$4$}
    }
    child { node {$\{1,5\}$}
      child { node {$\{2,4\}$}
        edge from parent node[fill=white] {$1$}
      }
      child { node {$\{3,8\}$}
        edge from parent node[fill=white] {$5$}
      }
      edge from parent node[fill=white] {$7$}
    }
    edge from parent node[fill=white] {$6$}
  };

\end{tikzpicture}
\end{center}
\caption{Tree representation of $\mathcal{\widehat{F}}$}
\label{fig_3}
\end{figure}

The main benefit of using representatives in the tree form stems from
the fact that their sizes are bounded by a function of only $p$ and $q$
(i.e., maximum size of a representative does not depend on $|\mathcal{F}|$).
It gives a way of storing paths of given length
between two vertices of a graph
in a space-efficient way (see Appendix for details).

The algorithm for finding a cycle of given length has two stages.
In the first stage, vertices of \emph{high degree} are processed to determine
if any of them belongs to a cycle. This is realized using algorithm \texttt{cycleThrough}
from Lemma \ref{cycle_through_v_lemma}. Since there are not many vertices of \emph{high degree},
this can be realized efficiently.

In the second stage, we remove vertices of \emph{high degree} from the graph.
Then, we group all simple paths that are half the cycle length long
by their endpoints and compute representatives for every such set
(see Lemma \ref{representative_creation}).
For each pair of vertices
$(u, v)$, we determine (using \texttt{findDisjoint} from
Lemma \ref{representative_find_empty_intersection_lemma})
if there are two paths: $p$ from $u$ to $v$ and $p'$ from $v$ to $u$, such that
$p$ and $p'$ do not share any internal vertices. If this is the case, $C = p\cup p'$ is
a cycle of required length.

The following representatives-related lemmas are proved in the Appendix.

\begin{lemma} \label{representative_creation}
We can compute a $q$-representative $\widehat{\mathcal{F}}$ for a collection of $p$-sets $\mathcal{F}$,
of size $|\mathcal{\widehat{F}}| \leq \sum_{i=1}^{q}p^i$,
in $O\Big(\big(\sum_{i=1}^{q+1}p^i\big)\cdot scan\big(|\mathcal{F}|\big)\Big)$ I/O's.
\end{lemma}

\begin{lemma} \label{representative_find_empty_intersection_lemma}
For a collection of $p$-sets, $\mathcal{F}$,
and a collection of $q$-sets, $\mathcal{G}$,
there is an external-memory method,
\verb\findDisjoint\$(\mathcal{F}, \mathcal{G})$,
that returns a pair of sets $(A,B)$
($A \in \mathcal{F}$, $B \in \mathcal{G}$) s.t. $A \cap B = \emptyset$
or returns $\epsilon$
if there are no such pairs of sets.
\verb\findDisjoint\ uses
$O\Big(\big(\sum_{i=1}^{q+3}p^i + \sum_{i=1}^{p+3}q^i\big)\cdot scan\big(|\mathcal{F}| + |\mathcal{G}|\big)\Big)$
I/O's.
\end{lemma}

\begin{lemma} \label{cycle_through_v_lemma}
Let $G=(V,E)$. A cycle of length exactly $k$ that passes through arbitrary
$v \in V$, if it exists, can be found by an external-memory algorithm \verb\cycleThrough\$(G, k, v)$
in $O\big((k-1)!\cdot scan(m)\big)$ I/O's, where $m = |E|$, via the use of representatives.
\end{lemma}

Before we present our result for naturally sparse graphs, we first give
an external-memory method for general graphs.

\begin{theorem} Let $G=(V,E)$ be a directed or an undirected graph. There is an external-memory
algorithm that decides if $G$ contains a cycle of length exactly $c \in \{2k-1, 2k\}$,
and finds such cycle if it exists, that takes
$O\big(k\cdot sort(m^{2-\frac{1}{k}}) + (2k-1)!\cdot scan(m^{2-\frac{1}{k}})\big)$ I/O's.
\end{theorem}
\begin{proof}

Algorithm \ref{alg_general_cycles} handles the case of general graphs
(which are not necessarily \emph{naturally sparse}), and cycles
of length $c=2k$ (the case of $c=2k-1$ is analogous).

\begin{algorithm}
\caption{Short cycles in general graphs}
\label{alg_general_cycles}
\begin{algorithmic}[1]
\STATE $\Delta \leftarrow m^\frac{1}{k}$
\FORALL{$v$ -- vertex of degree $\geq \Delta$}
  \STATE $C \leftarrow \texttt{cycleThrough}(G, 2k, v)$
  \IF{$C \neq \epsilon$}
    \RETURN $C$
  \ENDIF{}
\ENDFOR
\STATE remove vertices of degree $\geq \Delta$ from $G$
\STATE generate all directed paths of length $k$ in $G$
\STATE sort the paths lexicographically, according to their endpoints
\STATE group all paths $u \leadsto v$ into collection of $(k-1)$-sets $\mathcal{F}_{uv}$
\FORALL{pairs $(\mathcal{F}_{uv}, \mathcal{F}_{vu})$}
  \STATE $P \leftarrow \texttt{findDisjoint}(\mathcal{F}_{uv}, \mathcal{F}_{vu})$
  \IF{$P = (A, B)$}
    \RETURN $C = A \cup B$
  \ENDIF
\ENDFOR
\RETURN $\epsilon$
\end{algorithmic}
\end{algorithm}

Since there are at most $m/\Delta = m^{1-\frac{1}{k}}$ vertices of degree at least $\Delta$, and
each call to \texttt{cycleThrough} requires
$O\big((2k-1)!\cdot scan(m)\big)$ I/O's (by Lemma \ref{cycle_through_v_lemma}),
the first \texttt{for} loop (lines 2-7) takes $O\big(m^{1 - \frac{1}{k}}\cdot(2k-1)!\cdot scan(m)\big) =
O\big((2k-1)!\cdot scan(m^{2-\frac{1}{k}})\big)$ I/O's.

Removing vertices of high degree in line 8 is realized just like line 5
of Algorithm \ref{alg_degorder}, in $O\big(sort(m)\big)$ I/O's.
There are at most $m\Delta^{k-1} = m^{2-\frac{1}{k}}$ paths to be generated in line 9.
It can be done in $O\big(k\cdot sort(m^{2-\frac{1}{k}})\big)$ I/O's (see Appendix).
Sorting the paths (line 10) takes $O\big(sort(m^{2-\frac{1}{k}})\big)$ I/O's.
After that, creating $\mathcal{F}_{uv}$'s (line 11) requires $O\big(scan(m^{2-\frac{1}{k}})\big)$ I/O's.

The \texttt{groupF} procedure groups $\mathcal{F}_{uv}$ and $\mathcal{F}_{vu}$ together.
Assume we store $\mathcal{F}_{uv}$'s as tuples $(u, v, S)$, for $S \in \mathcal{F}_{uv}$, in a list $F$.
By $u\prec v$ we denote that $u$ precedes $v$ in an arbitrary ordering of $V$.
For $u\prec v$, tuples $(u, v, 1, S)$ from line 3 mean that $S \in \mathcal{F}_{uv}$,
while tuples $(u, v, 2, S)$ from line 5 mean that $S \in \mathcal{F}_{vu}$.
The \texttt{for} loop (lines 1-7) clearly takes $O\big(scan(m^{2-\frac{1}{k}})\big)$ I/O's.
After sorting $F$ (line 8) in $O\big(sort(m^{2-\frac{1}{k}})\big)$ I/O's, tuples for sets from
$\mathcal{F}_{uv}$ directly precede those for sets from $\mathcal{F}_{vu}$, allowing us to execute
line 9 in $O\big(scan(m^{2-\frac{1}{k}})\big)$ I/O's.

\vspace{\baselineskip}
\noindent
\fbox{
  \parbox{\linewidth}{
    \begin{algorithmic}[1]
    \REQUIRE \texttt{groupF}
    \FORALL{$(u, v, S)$ in $F$}
      \IF{$u\prec v$}
        \STATE write $(u, v, 1, S)$ back to $F$
      \ELSE
        \STATE write $(v, u, 2, S)$ back to $F$
      \ENDIF
    \ENDFOR
    \STATE sort $F$ lexicographically
    \STATE scan $F$ to determine pairs $(\mathcal{F}_{uv}, \mathcal{F}_{uv})$
    \end{algorithmic}
  }
}
\vspace{\baselineskip}

Based on Lemma \ref{representative_find_empty_intersection_lemma},
the total number of I/O's in calls to \texttt{findDisjoint}
in Algorithm \ref{alg_general_cycles}, line 13 is
\begin{eqnarray*}
& O\Big(\sum_{u,v}\big(\sum_{i=1}^{k+2}(k-1)^i\cdot scan(|\mathcal{F}_{uv}| + |\mathcal{F}_{vu}|)\big)\Big) & \\
& = O\Big(\big(\sum_{i=1}^{k+2}(k-1)^i\big)\cdot\sum_{u,v}scan\big(|\mathcal{F}_{uv}| + |\mathcal{F}_{vu}|\big)\Big) & \\
& = O\big((2k-1)!\cdot scan(m^{2-\frac{1}{k}})\big) &
\end{eqnarray*}
as we set $p=q=k-1$ and
$\sum_{i=1}^{k+2}(k-1)^i = O\big((k-1)^{k+3}\big) = O\big((2k-1)!\big)$.

Putting it all together, we get that Algorithm \ref{alg_general_cycles} runs in
$O\big(sort(m^{2-\frac{1}{k}}) + (2k-1)!\cdot scan(m^{2-\frac{1}{k}})\big)$ total I/O's.
\qed
\end{proof}

\begin{theorem}
Let $G=(V,E)$ be a directed or an undirected graph.
There is an external-memory algorithm
that, given $L$ -- a $\delta$-degeneracy ordering of $G$
(for $\delta = (2+\epsilon)d$), finds
a cycle of length exactly $c$, or concludes that it does not exist:
\begin{list}{\labelitemi}{\leftmargin=1em}
\item[(i)] in $O\Big(\delta^{1-\frac{1}{k}}\cdot \big(k\cdot sort(m^{2-\frac{1}{k}})+(4k)!\cdot scan(m^{2-\frac{1}{k}})\big)\Big)$
I/O's if $c = 4k-2$
\item[(ii)] in $O\Big(\delta\cdot \big(k\cdot sort(m^{2-\frac{1}{k}})+(4k)!\cdot scan(m^{2-\frac{1}{k}})\big)\Big)$
I/O's if $c = 4k-1$ or $c = 4k$
\item[(iii)] in $O\Big(\delta^{1+\frac{1}{k}}\cdot \big(k\cdot sort(m^{2-\frac{1}{k}})+(4k)!\cdot scan(m^{2-\frac{1}{k}})\big)\Big)$
I/O's if $c = 4k+1$
\end{list}
\end{theorem}
\begin{proof}
We describe the algorithm for the case of directed $G$, with $c = 4k+1$, as
other cases are similar (and a little easier).
We assume that $\delta < m^\frac{1}{2k+1}$, which is obviously the case
for \emph{naturally sparse} graphs.
Otherwise, running Algorithm \ref{alg_general_cycles} on $G$
achieves the advertised complexity.

\begin{algorithm}
\caption{Short cycles in degenerate graphs}

\label{alg_degenerate_cycles}
\begin{algorithmic}[1]
\STATE $\Delta \leftarrow m^\frac{1}{k}/\delta^{1+\frac{1}{k}}$
\FORALL{$v$ -- vertex of degree $\geq \Delta$}
  \STATE $C \leftarrow \texttt{cycleThrough}(G, 4k+1, v)$
  \IF{$C \neq \epsilon$}
    \RETURN $C$
  \ENDIF
\ENDFOR
\STATE remove vertices of degree $\geq \Delta$ from $G$
\STATE generate directed paths of length $2k$ and $2k+1$ in $G$
\STATE sort the paths lexicographically, according to their endpoints
\STATE group all paths $u \leadsto v$ of length $2k$ into collection of $(2k-1)$-sets $\mathcal{F}_{uv}$
\STATE group all paths $u \leadsto v$ of length $2k+1$ into collection of $(2k)$-sets $\mathcal{G}_{uv}$
\FORALL{pairs $(\mathcal{F}_{uv}, \mathcal{G}_{uv})$}
  \STATE $P \leftarrow \texttt{findDisjoint}(\mathcal{F}_{uv}, \mathcal{G}_{vu})$
  \IF{$P = (A, B)$}
    \RETURN $C = A \cup B$
  \ENDIF
\ENDFOR
\RETURN $\epsilon$
\end{algorithmic}
\end{algorithm}

Algorithm \ref{alg_degenerate_cycles} is remarkably similar to Algorithm \ref{alg_general_cycles} and
so is its analysis.
Differences lie in the value of $\Delta$ and in line 9, when only \emph{some} paths of length $2k$ and $2k+1$
are generated.
As explained in \cite{short_cycles}, it suffices to only consider all $(2k+1)$-paths that start with
two backward-oriented (in $L$) edges and all $2k$-paths that start with a backward-oriented (in $L$) edge.
The number of these paths is $O(m^{2-\frac{1}{k}}\delta^{1+\frac{1}{k}})$. Since we can generate them
in $O\left(k\delta^{1+\frac{1}{k}}\cdot sort(m^{2-\frac{1}{k}})\right)$ I/O's (see Appendix),
and there are at most $O(m^{1-\frac{1}{k}}\delta^{1+\frac{1}{k}})$ vertices in $G$ of degree $\geq \Delta$,
the theorem follows. \qed
\end{proof}

\section{All Maximal Cliques} \label{max_cliques_section}

The Bron-Kerbosch algorithm \cite{bron73} is often the choice when one needs to list
all maximal cliques of an undirected graph $G=(V,E)$.
It was initially improved by Tomita {\it et al.}~\cite{tomita2006}.
We present this improvement as the \texttt{BronKerboschPivot} procedure
($\Gamma(v)$ denotes the set of neighbors of vertex $v$).

\vspace{\baselineskip}
\noindent
\fbox{
  \parbox{\linewidth}{
    \begin{algorithmic}[1]
    \REQUIRE \texttt{BronKerboschPivot($P$, $R$, $X$)}
    \IF{$P \cup X = \emptyset$}
      \PRINT $R\qquad$ \emph{//maximal clique}
    \ENDIF
    \STATE $u \leftarrow$ vertex from $P\cup X$ that maximizes $|P\cap\Gamma(u)|$
    \FORALL{$v \in P\setminus \Gamma(v)$}
      \STATE \texttt{BronKerboschPivot($P\cap\Gamma(v)$, $R\cup\{v\}$, $X\cap\Gamma(v)$)}
      \STATE $P \leftarrow P\setminus\{v\}$
      \STATE $X \leftarrow X\cup\{v\}$
    \ENDFOR
    \end{algorithmic}
  }
}
\vspace{\baselineskip}

\noindent
The meaning of the arguments to \texttt{BronKerboschPivot}:
$R$ is a (possibly non-maximal) clique, $P$ and $X$ are a division of the set of vertices that are neighbors
of all vertices in $R$, s.t. vertices in $P$ are to be considered for adding to $R$ while vertices in $X$
are restricted from the inclusion.

Whereas Tomita {\it et al.} run the algorithm as
\texttt{BronKerboschPivot($V$,$\emptyset$,$\emptyset$)},
Eppstein {\it et al.}~\cite{els2010} improved it even further for the case of a $d$-degenerate $G$
by utilizing its $d$-degeneracy ordering $L = \{v_1, v_2, \ldots, v_n\}$
and by performing $n$ independent calls to \texttt{BronKerboschPivot}.
Algorithm \ref{alg_max_cliques} presents their version. It runs in time $O(dn3^{d/3})$ in the RAM model.

\begin{algorithm}
\caption{Maximal cliques in degenerate graph}

\label{alg_max_cliques}
\begin{algorithmic}[1]
\FOR{$i\leftarrow 1\ldots n$}
  \STATE $P\leftarrow \Gamma(v_i)\cap \{v_j\colon j > i\}$
  \STATE $X\leftarrow \Gamma(v_i)\cap \{v_j\colon j < i\}$
  \STATE \texttt{BronKerboschPivot($P$,$\{v_i\}$,$X$)}
\ENDFOR
\end{algorithmic}
\end{algorithm}

The idea behind Algorithm \ref{alg_max_cliques} is to limit the depth of
recursive calls to $|P| \leq d$ and then apply the analysis
of Tomita {\it et al.} \cite{tomita2006}.

We show how to efficiently implement Algorithm \ref{alg_max_cliques} in the
external memory model using a $(2+\epsilon)d$-degeneracy ordering of $G$.
Following \cite{els2010}, we define subgraphs $H_{P,X}$ of $G$.

\begin{definition}[Graphs $H_{P,X}$]
Subgraph $H_{P,X}=(V_{P,X},E_{P,X})$ of $G=(V,E)$ is defined as follows:
\begin{eqnarray*}
V_{P,X} & = & P\cup X \\
E_{P,X} & = & \{(u,v)\colon (u,v)\in E\land (u\in P \lor v\in P)\}
\end{eqnarray*}
\end{definition}

That is, $H_{P,X}$ contains all edges in $G$ whose endpoints are from $P\cup X$, and
at least one of them lies in $P$.
To ensure efficiency, $H_{P,X}$ is passed as an additional argument
to every call to \texttt{BronKerboschPivot} with $P$ and $X$.
It is used in determining $u$ at line 4 of \texttt{BronKerboschPivot}
(we simply choose a vertex of highest degree in $H_{P,X}$).

The following two lemmas regarding construction of $H_{P,X}$'s are proved in the Appendix.

\begin{lemma} \label{alg_max_cliques_lemma_1}
Given a $\delta$-degeneracy ordering $L$ of an undirected $d$-degenerate graph $G$ ($\delta = (2+\epsilon)d$),
all initial sets $P$, $X$, and graphs $H_{P,X}$ that are passed to \verb\BronKerboschPivot\
in line 4 of Algorithm \ref{alg_max_cliques} can be generated in $O(sort(\delta^2n))$ I/O's.
\end{lemma}

\begin{lemma} \label{alg_max_cliques_lemma_2}
Given a $\delta$-degeneracy ordering $L$ of an undirected $d$-degenerate graph $G$ ($\delta = (2+\epsilon)d$),
in a call to \verb\BronKerboschPivot\ that was given $H_{P,X}$, with $|P| = p$ and
$|X| = x$, all graphs $H_{P\cap\Gamma(v),X\cap\Gamma(v)}$ that have to be passed to recursive
calls in line 6, can be formed in $O(sort(\delta p^2(p+x)))$ I/O's.
\end{lemma}

\begin{theorem}
Given a $\delta$-degeneracy ordering $L$ of an undirected $d$-degenerate graph $G$ ($\delta = (2+\epsilon)d$),
we can list all its maximal cliques in $O(3^{\delta/3}sort(\delta n))$ I/O's.
\end{theorem}

\begin{proof}
Consider a call to \texttt{BronKerboschPivot($P_v$, $\{v\}$, $X_v$)}, with $|P_v| = p$ and $|X_v| = x$.
Define $\widehat{D}(p,x)$ to be the maximum number of I/O's in this call.
Based on Lemma \ref{alg_max_cliques_lemma_2}, $\widehat{D}(p,x)$ satisfies the following recurrence relation:
\begin{displaymath}
\widehat{D}(p,x) \leq \left\{
  \begin{array}{ll}
    \max_k\{k\widehat{D}(p-k,x)\} + O\big(sort(\delta p^2(p+x))\big) & \quad \textrm{if}\ p>0 \\
    e &\quad \textrm{if}\ p = 0
  \end{array}
\right.
\end{displaymath}
for constant $e$ greater than zero, which can be rewritten as
\begin{displaymath}
\widehat{D}(p,x) \leq \left\{
  \begin{array}{ll}
    \max_k\{k\widehat{D}(p-k,x)\} + c\cdot\frac{\delta p^2(p+x)}{DB}\log_{M/B}(\delta p^2(p+x)) & \quad \textrm{if}\ p>0 \\
    e &\quad \textrm{if}\ p = 0
  \end{array}
\right.
\end{displaymath}

for a constant $c>0$.
Since $p\leq \delta$ and $p+x\leq n$, we have $\log_{M/B}(\delta p^2(p+x)) \leq \log_{M/B}(\delta^3n) = O(\log_{M/B}n)$
for $\delta=O(1)$.
Thus, the relation for $\widehat{D}(p,x)$:
\begin{displaymath}
\widehat{D}(p,x) \leq \left\{
  \begin{array}{ll}
    \max_k\{k\widehat{D}(p-k,x)\} + \delta p^2(p+x)\cdot \frac{c'\log_{M/B}n}{DB} & \quad \textrm{if}\ p>0 \\
    e &\quad \textrm{if}\ p = 0
  \end{array}
\right.
\end{displaymath}
where $c'$ and $e$ are constants greater than zero. Note that this is the relation
for $D(p,x)$ of Eppstein {\it et. al}~\cite{els2010}
(we set $d = \delta$, $c_1=\frac{c'\log_{M/B}n}{DB}$ and $c_2 = e$).
Since the solution for $D(p,x)$ was $D(p,x) = O((d+x)3^{p/3})$, the solution for $\widehat{D}(p,x)$ is
\begin{displaymath}
\widehat{D}(p,x) = O\Big((\delta+x)3^{p/3}\cdot \frac{\scriptstyle{c'\log_{M/B}n}}{\scriptstyle{DB}}\Big) =
O\Big(\frac{\scriptstyle{\delta+x}}{\scriptstyle{DB}}3^{p/3}\log_{M/B}n\Big)
\end{displaymath}

The total size of all sets $X_v$ passed to initial calls to \texttt{BronKerboschPivot} is
$O(\delta n)$, and every set $P$ has at most $\delta$ vertices.
It follows that the total number of I/O's in recursive calls is
\begin{displaymath}
\sum_vO\Big(\frac{\scriptstyle{\delta+|X_v|}}{\scriptstyle{DB}}3^{\delta/3}\log_{M/B}n\Big) =
O\Big(3^{\delta/3}\frac{\scriptstyle{\delta n}}{\scriptstyle{DB}}\log_{M/B}n\Big) = O\big(3^{\delta/3}sort(\delta n)\big)
\end{displaymath}

Combining this with Lemma \ref{alg_max_cliques_lemma_1}, we get that our external memory
version of Algorithm \ref{alg_max_cliques} takes $O\big(sort(\delta^2n) + 3^{\delta/3}sort(\delta n)\big) =
O\big(3^{\delta/3}sort(\delta n)\big)$ I/O's. \qed
\end{proof}

\bibliographystyle{abbrv}
\bibliography{refs}

\appendix
\newpage

\section{Appendix}

\subsection{Proof of Lemma \ref{alg_degorder_lemma_2}}

The input is a $d$-degenerate graph $G = (V,E)$. Let $V = (1, \ldots, n)$.
We store $E$ as set of edges $(u,v)$.
We assume an order on $V$ that we utilize during sorting.

Finding vertices of smallest degree (Algorithm \ref{alg_degorder}, line 3)
is realized by the \texttt{smallVertices} procedure.

\vspace{\baselineskip}
\noindent
\fbox{
  \parbox{\linewidth}{
    \begin{algorithmic}[1]
    \REQUIRE{\texttt{smallVertices}}
    \FORALL{$v$ -- vertex of $G$}
      \STATE $d(v) \leftarrow\ $ degree of $v$
      \STATE store pair $(d(v),v)$ in $F$
    \ENDFOR
    \STATE sort $F$ lexicographically
    \STATE $S \leftarrow $ first $n\epsilon/(2+\epsilon)$ vertices in $F\qquad$ \emph{// of smallest degree}
    \end{algorithmic}
  }
}
\vspace{\baselineskip}

\noindent
Computing degrees of vertices in the \texttt{for} loop (lines 1-4)
is easily realized in $O(sort(dn))$ I/O's.
First, $E$ is sorted lexicographically in $O(sort(dn))$ I/O's.
After that, edges of $E$ form blocks ordered
by their starting vertex, so a simple scan taking $O(scan(dn))$ I/O's is enough
to determine degrees of the vertices.

Sorting $F$ (line 5) is done in $O(sort(n))$ I/O's. After that, $S$ is just $n\epsilon/(2 + \epsilon)$
first items of $F$ and its construction (line 6) takes $O(scan(n\epsilon/(2 + \epsilon)))$ I/O's.
Likewise, appending $S$ to $L$ (Algorithm \ref{alg_degorder}, line 4)
takes $O(scan(n\epsilon/(2 + \epsilon)))$ I/O's.

Finally, removing edges adjacent to $S$ from $G$ (Algorithm \ref{alg_degorder}, line 5) is realized as follows.

\vspace{\baselineskip}
\noindent
\fbox{
  \parbox{\linewidth}{
    \begin{algorithmic}[1]
    \STATE sort $E$ lexicographically
    \FORALL{$(u,v)$ -- edge in $E$}
      \IF{$u \in S$}
        \STATE add tuples $(u, v, \texttt{"-"})$ and $(v, u, \texttt{"-"})$ to $E$
      \ENDIF
    \ENDFOR
    \STATE sort $E$ lexicographically
    \FORALL{$p$, $q$ -- consecutive tuples in $E$}
      \IF{$p = (u, v)$ and $q = (u, v, \texttt{"-"})$}
        \STATE do not write $p$ back to $E$
      \ELSIF{$p = (u, v)$}
        \STATE write $p$ back to $E$
      \ELSE[$\qquad// p = (u,v, \texttt{"-"})$]
        \STATE do not write $p$ back to $E$
      \ENDIF
    \ENDFOR
    \end{algorithmic}
  }
}
\vspace{\baselineskip}

\noindent
Sorting $E$ in line 1 takes $O(sort(dn))$ I/O's.
The first \texttt{for} loop (lines 2-6) takes $O(scan(dn))$ I/O's (vertices in $S$ are stored
according to the order on $V$, in the same relative order
as the origins of edges in $E$, so it is realized by a single synchronized scan
going through $E$ and $S$ at the same time).
Each edge in $E$ causes at most 2 tuples to be added to $E$ in line 4,
so clearly the size of $E$ is $O(dn)$ after line 6.
Sorting $E$ (line 7) obviously takes time $O(sort(dn))$.

The last \texttt{for} loop (lines 8-16) is easily realized by a single scan of $E$, in
$O(scan(dn))$ I/O's. Correctness follows from two facts. First, for each edge $(u,v)$, if $u\in S$,
the tuple $(u,v,\texttt{"-"})$ (meaning that $(u,v)$ does not belong to $G$ after this iteration)
is added to $E$ in line 4. Second, if $(u,v, \texttt{"-"})$ is in $E$, its direct predecessor
is $(u,v)$ (or another copy of $(u,v, \texttt{"-"})$) after $E$ was sorted lexicographically in line 7.
Therefore, the edges that are no longer in $G$ are rejected in line 10. Also, no tuples
$(u,v,\texttt{"-"})$ are further stored in $E$ (line 14).

Altogether, lines 3-5 of Algorithm \ref{alg_degorder} are implemented in $O(sort(dn))$ I/O's. \qed

\subsection{Representatives}

Monien~\cite{monien_paths} gave
a simple algorithm (which we call \texttt{repQuery}) that
operates on representatives in the tree form described in Sec. \ref{cycles_section}.
Given $\widehat{\mathcal{F}}$ --
a representative for $\mathcal{F}$, \texttt{repQuery($\widehat{\mathcal{F}}$, $B$)}
decides for a $q$-set $B$ whether there exists a set $A \in \mathcal{F}$ s.t.
$A \cap B = \emptyset$, and returns such $A$ if it exists. The running time of
\texttt{repQuery} is $O(pq)$ in the RAM model.
We use \texttt{repQuery} in our algorithms "as is", i.e.,
we allow it to take $O(pq)$ I/O's.

\subsubsection{Proof of Lemma \ref{representative_creation}}

The size of the resulting tree is is bounded by $\sum_{i=1}^{q}p^i$
(number of nodes in a $p$-ary tree of height $q$).
We can afford to build the tree one node at a time,
spending $O\big(scan(p|\mathcal{F}|)\big)$ I/O's on each node.
Procedure \texttt{repLabel} labels vertex $v$ in the representative tree for $\mathcal{F}$
in $O\big(scan(p|\mathcal{F}|)\big)$ I/O's.

\vspace{\baselineskip}
\noindent
\fbox{
  \parbox{\linewidth}{
    \begin{algorithmic}[1]
    \REQUIRE \texttt{repLabel($\mathcal{F}$, $v$)}
    \STATE $\widehat{A} \leftarrow$ set $A$ in $\mathcal{F}$ s.t. $A\cap E(v)=\emptyset$
    \IF{$\widehat{A} \neq \epsilon$}
      \STATE label $v$ with $\widehat{A}$
      \STATE create children of $v$
      \STATE label edges from $v$ to its children with elements from $\widehat{A}$
    \ELSE
      \STATE label $v$ with $\lambda$
    \ENDIF
    \end{algorithmic}
  }
}
\vspace{\baselineskip}

\noindent
Set $\widehat{A}$ in line 1 can simply be found by scanning $\mathcal{F}$
in $O\big(scan(p|\mathcal{F}|)\big)$ I/O's ($p|\mathcal{F}|$ is the total size of
all $p$-sets in $\mathcal{F}$). Creating children of $v$ (line 4) and labeling
their edges (line 5) takes $O(p)$ I/O's.

Therefore, we compute a $q$-representative for $\mathcal{F}$ in
$O\Big(\big(\sum_{i=1}^q p^i\big)\cdot scan\big(p|\mathcal{F}|\big)\Big) =
O\Big(\big(\sum_{i=1}^q p^{i+1}\big)\cdot scan\big(|\mathcal{F}|\big)\Big) =
O\Big(\big(\sum_{i=1}^{q+1} p^i\big)\cdot scan\big(|\mathcal{F}|\big)\Big)$ I/O's.
\qed

\subsubsection{Proof of Lemma \ref{representative_find_empty_intersection_lemma}}
The proof is essentially the same as that of Lemma 3.2 in \cite{short_cycles}:
first we compute a $q$-representative $\widehat{\mathcal{F}}$ of $\mathcal{F}$,
in $O\Big(\big(\sum_{i=1}^{q+1}\big)\cdot scan\big(|\mathcal{F}|\big)\Big)$ I/O's,
and a $p$-representative $\widehat{\mathcal{G}}$ of $\mathcal{G}$,
in $O\Big(\big(\sum_{i=1}^{p+1}\big)\cdot scan\big(|\mathcal{G}|\big)\Big)$ I/O's.

The sizes of $\widehat{\mathcal{F}}$ and $\widehat{\mathcal{G}}$ are bounded by
$\sum_{i=1}^q p^i$ and $\sum_{i=1}^p q^i$, respectively.
Assuming $p\geq q$ (w.l.o.g.), determining whether
$\widehat{\mathcal{F}}$ and $\widehat{\mathcal{G}}$ contain
two disjoint sets can be easily done in $O\big(\sum_{i=1}^{q+1}p^i\cdot pq\big) =
O\big(\sum_{i=1}^{q+1}p^{i+2}\big) = O\big(\sum_{i=1}^{q+3}p^i\big)$ I/O's,
by querying $\widehat{\mathcal{G}}$ (via \texttt{repQuery}) with all sets from $\widehat{\mathcal{F}}$.
\qed

\subsubsection{Proof of Lemma \ref{cycle_through_v_lemma}}

The (very) big picture of the \texttt{cycleThrough($G$, $k$, $v$)} algorithm is as follows:

\vspace{\baselineskip}
\noindent
\fbox{
  \parbox{\linewidth}{
    \begin{algorithmic}[1]
    \REQUIRE \texttt{cycleThrough($G$, $k$, $v$)}
    \FORALL{$u$ s.t. $(v, u) \in E$}
      \IF{there exists simple path $p$: $u \leadsto v$ of length exactly $k - 1$} \label{alg_cycle_through_v_if}
        \RETURN $p \cup (v,u)$
      \ENDIF
    \ENDFOR
    \RETURN $\epsilon$
    \end{algorithmic}
  }
}
\vspace{\baselineskip}

\noindent
Obviously, main difficulty lies in checking the condition in line \ref{alg_cycle_through_v_if}.
To show how we answer that query, let us first explain how \cite{monien_paths}
handles the paths.

Let $\mathcal{P}_{uv}^p$ denote the set of all simple paths from $u$ to $v$ of length exactly $(p + 1)$
(so that these paths have $p$ inner vertices).
First, all paths from $u$ to $v$ of length $(p + 1)$ that have the exact same set of inner vertices
are represented as a single set containing these vertices (the vertices are stored as
one of the paths; it provides a representative of the set).
Performing this compression on $\mathcal{P}_{uv}^p$ yields $\mathcal{F}_{uv}^p$ -- a family of $p$-sets:
\begin{displaymath}
\mathcal{F}_{uv}^p = \{S\colon S \,\textrm{is a set of inner vertices on some path from}\, u\, \textrm{to}\,
v\, \textrm{of length}\, p+1 \}
\end{displaymath}
The condition from line \ref{alg_cycle_through_v_if} of \texttt{cycleThrough}
is therefore equivalent to $\mathcal{F}_{uv}^{k-2}$ being nonempty. We will now focus on how to test if
this is the case.

The clou of \cite{monien_paths} was that having $q$-representatives for $\mathcal{F}_{uv}^p$
(for all $u \in V$) enables efficient computation of $(q-1)$-representatives
for $\mathcal{F}_{uv}^{p+1}$ (for all $u \in V$).
The labels for a $(q-1)$-representative tree for $\mathcal{F}^{p+1}_{uv}$ are computed
node by node. The algorithm is based on the following observation ($\gamma$ is the node
whose label we compute, $E(\gamma)$ is the set of edge labels on the way from $\gamma$ to root):
\begin{displaymath}
\begin{array}{c}
\exists U \in \mathcal{F}_{uv}^{p+1} \ \textrm{s.t.} \ U\cap E(\gamma) = \emptyset \\
\iff \\
\exists w \in V \setminus \{u,v\} \ \textrm{s.t.} \ (u,w) \in E \ \land \ w \notin E(\gamma) \ \land \\
\exists \widehat{U} \in \mathcal{F}_{wv}^p \ \textrm{s.t.} \ \widehat{U} \cap
(E(\gamma) \cup \{u\}) = \emptyset
\end{array}
\end{displaymath}
Having a $q$-representative for $\mathcal{F}_{wv}^p$ allows us to find $\widehat{U}$ (or determine
that it does not exist) via the \texttt{repQuery} algorithm.
Determining the label for $\gamma$ is therefore realized as follows:

\vspace{\baselineskip}
\noindent
\fbox{
  \parbox{\linewidth}{
    \begin{algorithmic}[1]
    \FORALL{$w$ s.t. $(u,w) \in E$ and $w \notin E(\gamma)$}
      \STATE $\widehat{U} \leftarrow$ \texttt{repQuery}$(\mathcal{F}_{wv}^p,\, E(\gamma) \cup \{u\})$
      \IF{$\widehat{U} \neq \epsilon$}
        \STATE label $\gamma$ with $\widehat{U} \cup \{w\}$
        \RETURN
      \ENDIF
    \ENDFOR
    \STATE label $\gamma$ with $\lambda$
    \end{algorithmic}
  }
}
\vspace{\baselineskip}

\texttt{repQuery} (querying a representative tree) takes $O(pq)$ I/O's, so
our implementation of labeling $\gamma$ takes
$O\big(pq\cdot scan\big(\Gamma(u)\big)\big)$ I/O's
(where $\Gamma(u)$ denotes the number of neighbors of $u$ in $G$).
It simply scans neighbors of $u$ and calls \texttt{repQuery} accordingly.
Because the $(q-1)$-representative tree for $\mathcal{F}_{uv}^{p+1}$ has size
bounded by $\sum_{i=1}^{q-1}(p+1)^i \leq q(p+1)^{q+1}$, labeling all its nodes requires
$O\Big(q(p+1)^{q-1}\cdot\big(pq\cdot scan\big(\Gamma(u)\big)\big)\Big) =
O\Big(q^2(p+1)^q\cdot scan\big(\Gamma(u)\big)\Big)$ I/O's.
We are computing $(q-1)$-representatives for
$\mathcal{F}_{uv}^{p+1}$'s for all $u$'s, so it takes
$O\Big(\sum_u\big(q^2(p+1)^q\cdot scan\big(|\Gamma(u)|\big)\big)\Big) =
O\Big(q^2(p+1)^q\cdot\sum_u scan\big(|\Gamma(u)|\big)\Big) = O\big(q^2(p+1)^q\cdot scan(m)\big)$
I/O's in total.

Our goal is to compute $0$-representatives $\widehat{\mathcal{F}}_{uv}^{k-2}$ for $\mathcal{F}_{uv}^{k-2}$
(for all $u\in V$). Then, by calling \texttt{repQuery($\widehat{\mathcal{F}}_{uv}^{k-2}$, $\emptyset$)},
we determine whether $\mathcal{F}_{uv}^{k-2}$ is nonempty, as it either returns $\epsilon$
(if $\mathcal{F}_{uv}^{k-2}$ is empty), or a set $A\in\mathcal{F}_{uv}^{k-2}$, representing
a path from $u$ to $v$ of length $k-1$.

We start with $(k-2)$-representatives for $\mathcal{F}_{uv}^0$'s. They are built as trees
having only the root vertex, labeled with either $\emptyset$ (if $(u,v)\in E$), or $\lambda$ (otherwise).
They can be obviously constructed in $O(scan(m))$ I/O's, via scanning $E$.
Based on the discussion above, computing $\widehat{\mathcal{F}}_{uv}^{k-2}$'s (for all $u\in V$) takes
\begin{eqnarray*}
& O\big(\sum_{p=0}^{k-3}(k-2-p)(p+1)^{k-2-p}\cdot scan(m)\big) & \\
& = O\big((k-2)\cdot scan(m)\cdot \sum_{p=0}^{k-3}(p+1)^{k-2-p}\big) & \\
& = O\big((k-2)\cdot scan(m)\cdot (k-2)!\big) & \\
& = O\big((k-1)!\cdot scan(m)\big) &
\end{eqnarray*}
I/O's, as it can be easily shown by induction that
$\sum_{p=0}^{k-3}(p+1)^{k-2-p} \leq (k-2)!$ for $k\geq6$.

\subsection{Path Generation}

\subsubsection{General Graphs (Algorithm \ref{alg_general_cycles})}

Recall that for Algorithm \ref{alg_general_cycles} we need to generate all directed paths of length $k$
in a graph $G=(V,E)$, where maximum degree of each vertex is bounded by $\Delta$. As shown in
\cite{short_cycles}, there are at most $O(m\Delta^{k-1}) = O(m^{2-\frac{1}{k}})$ such paths.
They are generated by the \texttt{pathGen} procedure.

\vspace{\baselineskip}
\noindent
\fbox{
  \parbox{\linewidth}{
    \begin{algorithmic}[1]
    \REQUIRE \texttt{pathGen}
    \STATE generate all sequences of length $k-1$, with elements from $\{1, \ldots, \Delta\}$
    \FORALL{$s$ -- sequence $\in \{1, \ldots, \Delta\}^{k-1}$}
      \STATE generate sequences $e|s$, for all $e \in E$
    \ENDFOR
    \FORALL{$s$ -- sequence $\in E\times\{1, \ldots, \Delta\}^{k-1}$}
      \STATE decode $s$ into $s' \in E^k$
      \IF{$s' \neq \epsilon$}
        \IF{$s'$ is a simple path}
          \PRINT $s'$
        \ENDIF
      \ENDIF
    \ENDFOR
    \end{algorithmic}
  }
}
\vspace{\baselineskip}

\noindent
Line 1 is simply realized in $O\big(scan(\Delta^{k-1})\big) = O\big(scan(m^{1-\frac{1}{k}})\big)$ I/O's.
Adding an edge at the beginning
of each sequence in the first \texttt{for} loop (lines 2-4) takes $O\big(m\cdot scan(\Delta^{k-1})\big)
= O\big(scan(m^{2-\frac{1}{k}})\big)$ I/O's.

Decoding a sequence $s \in E\times\{1, \ldots, \Delta\}^{k-1}$ into a path $s' \in E^k$
(\texttt{pathGen}, line 6) is conceptually straightforward.
$e$ is the first edge in the path. Then, each consecutive number $i$ determines next vertex on the path
-- $i$th neighbor of the previously decoded one (if it has less than $i$ neighbors,
the path is dropped as invalid).
The \texttt{decodePaths} procedure handles decoding of $S$ -- the set of sequences.
$s\texttt{[i]}$ is the $i$th element of tuple $s$, $s\texttt{[i].from}$ is the origin,
and $s\texttt{[i].to}$ is the destination vertex of edge at $s\texttt{[i]}$.

The first \texttt{for} loop (lines 1-3) decodes the starting edge of the sequence $s$ into two vertices, and then
cyclically shifts the resulting tuple by one position. It takes $O\big(scan(m^{2-\frac{1}{k}})\big)$ I/O's.
Each iteration of the second \texttt{for} loop (lines 4-12) decodes the next vertex of $s$ and
again cyclically shifts $s$ by one position.
Sorting $S$ in line 5 takes $O\big(sort(m^{2-\frac{1}{k}})\big)$ I/O's.
After that, the inner loop (lines 6-11) requires only $O\big(scan(m^{2-\frac{1}{k}})\big)$ I/O's
(it takes one synchronized scan of $S$ and $E$).
Invalid paths that do not meet the condition at line 8 are dropped.
Since the outer \texttt{for} loop runs for $k-1$ iterations,
\texttt{decodePaths} uses $O\big(k\cdot sort(m^{2-\frac{1}{k}})\big)$ I/O's.

\vspace{\baselineskip}
\noindent
\fbox{
  \parbox{\linewidth}{
    \begin{algorithmic}[1]
    \REQUIRE \texttt{decodePaths($S$)}
    \FORALL{$s$ -- sequence $\in E\times\{1,\ldots,\Delta\}^{k-1}$}
      \STATE write tuple $(s\texttt{[1].to}, s\texttt{[2]}, s\texttt{[3]}, \ldots,
        s\texttt{[k]}, s\texttt{[1].from})$ to $S$
    \ENDFOR
    \FORALL[\emph{$\quad$ // $S$ contains tuples
      $V\times\{1,\ldots,\Delta\}^{k-i}\times V^{i}$}]{$i\leftarrow1,\ldots, k-1$}
      \STATE sort $S$ lexicographically
      \FORALL{$s$ -- tuple in $S$}
        \STATE $u\leftarrow s\texttt{[2]}$th neighbor of $s\texttt[1]$ in $V$
        \IF{$u\neq\epsilon$}
          \STATE write tuple $(u, s\texttt{[3]}, s\texttt{[4]}, \ldots, s\texttt{[k+1]}, s\texttt{[1]})$ back to $S$
        \ENDIF
      \ENDFOR
    \ENDFOR
    \end{algorithmic}
  }
}
\vspace{\baselineskip}

\noindent

Verifying that a path is simple (\texttt{pathGen}, line 8) is done by checking that
it does not contain repeated vertices.
Thus, \texttt{pathGen} takes $O\big(sort(m^{2-\frac{1}{k}})\big)$ I/O's.

\subsubsection{Degenerate Graphs (Algorithm \ref{alg_degenerate_cycles})}

For Algorithm \ref{alg_degenerate_cycles}, we need to generate all paths of length $2k+1$ that
start with two backward-oriented (in $L$) edges. As shown in \cite{short_cycles}, there are at most
$O\big(m\sum_{i=0}^k\binom{2k}{i}\Delta^i\delta^{2k-i}\big) = O(2^{2k}m\Delta^k\delta^k)$ paths of length $2k+1$
in $G$. It follows from the fact that for each path $p$, either $p$ or $p^\mathrm{R}$ (the reverse
of $p$) has at most $k$ edges with opposite directions than in $L$.

The procedure \texttt{pathGenForward} generates paths $p$ that have at most $k$ edges with opposite
directions than in $L$.

\vspace{\baselineskip}
\noindent
\fbox{
  \parbox{\linewidth}{
    \begin{algorithmic}[1]
    \REQUIRE \texttt{pathGenForward}
    \STATE generate sequences $s = u|v$, for all $(u, v) \in E$
    \FOR{$i \leftarrow 1\ldots2k$}
      \FORALL{$s$ -- sequence $\in V^2\times
        (\{\texttt{"L"}, \texttt{"E"}\}\times\{1,\ldots,\Delta\})^{i-1}$}
        \STATE generate all sequences $s|(\texttt{"L"}, j)$, for $j \in \{1,\ldots,\delta\}$
        \IF{$s$ has $< k$ pairs of the type $(\texttt{"E"}, j)$}
          \STATE generate all sequences $s|(\texttt{"E"}, j)$, for $j \in \{1,\ldots,\Delta\}$
        \ENDIF
      \ENDFOR
    \ENDFOR
    \FORALL{$s$ -- sequence $\in V^2\times
      \big(\big\{\texttt{"L"}, \texttt{"E"}\big\}\times\big\{1,\ldots,\max\{\delta,\Delta\}\big\}\big)^{2k}$}
      \STATE decode $s$ into $s' \in E^{2k+1}$
      \IF{$s' \neq \epsilon$}
        \IF{$s'$ is a simple path}
          \PRINT $s'$
        \ENDIF
      \ENDIF
    \ENDFOR
    \end{algorithmic}
  }
}

\vspace{\baselineskip}
\noindent
The encoding of the sequences works as follows: it starts with two vertices, $u$ and $v$,
that represent the starting edge of the path.
$v$ is followed by $2k$ pairs of the format $(\texttt{"L"}, i)$ (with $i\in\{1,\ldots,\delta\}$)
or $(\texttt{"E"}, i)$ (with $i\in\{1,\ldots,\Delta\}$).
$(\texttt{"L"}, i)$ means that the next vertex is the $i$th neighbor of the current vertex in
degeneracy ordering $L$, and $(\texttt{"E"}, i)$ means that the next vertex is the $i$th
neighbor of the current vertex in $E$.
Any sequence $s$ has at most $k$ pairs
of the type $(\texttt{"E"}, j)$, and only edges represented by them may have
the opposite direction than in $L$.

The number of sequences generated by the first \texttt{for} loop (lines 2-9)
is clearly $O\big(m\sum_{i=0}^k\binom{2k}{i}\Delta^i\delta^{2k-i}\big) = O(2^{2k}m\Delta^k\delta^k)$,
and the whole generation process takes $O\big(scan(2^{2k}m\Delta^k\delta^k)\big)$ I/O's.
The sequences are decoded into paths in the second \texttt{for} loop (lines 10-17) in a manner
similar to \texttt{decodePaths}, using $O\big((2k+1)\cdot sort(2^{2k}m\Delta^k\delta^k)$ I/O's.

Paths of length $2k+1$, that have at most $k$ edges in the opposite direction than in $L$
when they are read backwards, can be generated by an analogous procedure \texttt{pathGenBackward}
(using $E^{\mathrm{R}}$ -- reversed edges instead of $E$), with
the same I/O complexity.

Procedure \texttt{pathGen2} generates all paths of length $2k+1$ that start with
two backward-oriented (in $L$) edges.

\vspace{\baselineskip}
\noindent
\fbox{
  \parbox{\linewidth}{
    \begin{algorithmic}[1]
    \REQUIRE \texttt{pathGen2}
    \STATE generate paths of length $2k-1$, via \texttt{pathGenForward}
    \STATE generate paths of length $2k-1$, via \texttt{pathGenBackward}
    \FORALL{$s$ -- generated path of length $2k-1$}
      \STATE generate all sequences $i|j|s$, for $(i,j) \in \{1,\ldots,\delta\}^2$
    \ENDFOR
    \FORALL{$s$ -- sequence $\in \{1,\ldots,d\}^2\times E^{2k-1}$}
      \STATE decode $s$ into $s' \in E^{2k+1}$
      \IF{$s' \neq \epsilon$}
        \IF{$s'$ is a simple path}
          \PRINT $s'$
        \ENDIF
      \ENDIF
    \ENDFOR
    \end{algorithmic}
  }
}

\vspace{\baselineskip}
Paths generated in lines 1-2 are the \emph{tails} of the resulting paths. The two numbers in sequences added
to these paths in line 4 denote the edges in $L$ that are to be taken to determine first two vertices
on the final path, starting at the tail's first vertex. This assures that these edges are backward-oriented
in $L$.

The number of paths generated by \texttt{pathGen2} is
$O(\delta^2\cdot2^{2k-2}m\Delta^{k-1}\delta^{k-1}) = O(2^{2k-2}m\Delta^{k-1}\delta^{k+1})
= O(2^{2k-2}m^{2-\frac{1}{k}}\delta^{1+\frac{1}{k}})$, and its I/O complexity is
$O\big(k\cdot sort(2^{2k-2}m^{2-\frac{1}{k}}\delta^{1+\frac{1}{k}})\big)$.

Since generating all paths of length $2k$ that begin with a backward-oriented (in $L$) edge
is essentially the same as \texttt{pathGen2}, path generation in this case
requires $O\big(k\cdot sort(2^{2k-2}m^{2-\frac{1}{k}}\delta^{1+\frac{1}{k}})\big)$ I/O's.

\subsection{Graph Reordering} \label{graph_reordering}
Assume we are given a graph $G=(V,E)$, with $V = (1,2,\ldots,n)$,
and $d$-degeneracy ordering $L=(v_1, v_2, \ldots, v_n)$ of $V$.
Our goal is to \emph{reorder} $G$ according to $L$, i.e.,
substitute each edge $(v_i, v_j)\in E$ with an edge $(i,j)$.

\vspace{\baselineskip}
\noindent
\fbox{
  \parbox{\linewidth}{
    \begin{algorithmic}[1]
    \REQUIRE \texttt{reorderG($G$, $L$)}
    \FORALL{$k\leftarrow 1,2$}
      \FORALL{$v_i$ -- $i$th vertex in $L$}
        \STATE append tuple $(v_i, \texttt{"$i$"})$ to $E$
      \ENDFOR
      \STATE sort $E$ lexicographically
      \FORALL{$p$ -- tuple in $E$}
        \IF{$p = (u, v)$}
          \STATE $q\leftarrow$ tuple $(u, \texttt{"$i$"})\qquad$ \emph{// precedes $p$}
          \STATE write $(v, i)$ back to $E$
        \ELSE
          \STATE do not write $(v, i)$ back to $E$
        \ENDIF
      \ENDFOR
    \ENDFOR
    \end{algorithmic}
  }
}
\vspace{\baselineskip}

\noindent
A single iteration of the outer \texttt{for} loop (lines 1-14) first renames origins of edges in $E$
and then reverts them (thus, after 2 iterations edges have their original directions).
First, it adds vertices along with their positions in $L$ to $E$ in lines 2-4. This takes $O(scan(n))$ I/O's.
Then it sorts E (line 5) in $O(sort(dn))$ I/O's. The next \texttt{for} loop (lines 6-13) scans
$E$, renames origins of edges to their positions in $L$ and outputs their opposite versions.
The tuples $q$ obtained in line 8 can clearly be found in the same single scan (tuple $(u, \texttt{"$i$"})$
directly precedes all edges $(u,v)$ after $E$ was sorted), so this part is done in $O(scan(dn))$ I/O's.

Therefore, \texttt{reorderG} runs in $O(sort(dn))$ I/O's altogether.

\subsection{Proof of Lemma \ref{alg_max_cliques_lemma_1}}
First, observe that the total size of all sets $P$ and $X$ passed to initial
calls to \texttt{BronKerboschPivot} is $O(m) = O(\delta n)$. To see this, note that every edge $(v_i, v_j)$
(with $v_i$ preceding $v_j$ in $L=(v_1, v_2, \ldots, v_n)$) puts $v_i$ into initial $X$ for
$R = \{v_j\}$ and $v_j$ into initial $P$ for $R = \{v_i\}$. Since the size of each $P$ is at
most $\delta$, the total number of edges in all $H_{P,X}$'s is $O(\delta^2n)$.

Our approach is to generate and store all $P$'s, $X$'s and $H_{P,X}$'s at the very beginning
of the algorithm, and then just pass
appropriate $P$, $X$, and $H_{P,X}$ to each initial call to \texttt{BronKerboschPivot}.

Procedure \texttt{genPX} generates all initial $P$'s and $X$'s in $O(sort(\delta n))$ I/O's.
We assume that for each vertex $v\in V$ we know its position in $L$ (i.e., we know $i$ for $v=v_i$).
Also, for each edge $(u,v)\in E$, we know positions in $L$ of its endpoints (i.e.,
we know both $i$ and $j$ for $(u,v) = (v_i,v_j)$).
We can easily achieve this in $O(sort(\delta n))$ I/O's (see Sec. \ref{graph_reordering}).
By $u \stackrel{\scriptscriptstyle{L}}{\prec} v$, we denote that $u$ precedes $v$ in $L$.
The \texttt{for} loop (lines 1-5) clearly takes $O(scan(\delta n))$ I/O's.

\vspace{\baselineskip}
\noindent
\fbox{
  \parbox{\linewidth}{
    \begin{algorithmic}[1]
    \REQUIRE \texttt{genPX($G$, $L$)}
    \FORALL{$(v_i, v_j)$ -- edge in $E$}
      \IF{$v_i \stackrel{\scriptscriptstyle{L}}{\prec} v_j$}
        \STATE output tuples $(v_i, \texttt{"P"}, v_j)$, and $(v_j, \texttt{"X"}, v_i)$ to set $F$
      \ENDIF
    \ENDFOR
    \STATE sort $F$ lexicographically
    \STATE scan $F$ to create sets $P$ and $X$
    \end{algorithmic}
  }
}
\vspace{\baselineskip}

\noindent

The meaning of $(v_i, \texttt{"P"}, v_j)$ is "add $v_j$ to $P_{v_i}$", and
of $(v_j, \texttt{"X"}, v_i)$: "add $v_i$ to $X_{v_j}$". The above discussion
explains why this information allows to generate all $P$'s and $X$'s.
Therefore, after set $F$ is sorted in $O(sort(\delta n))$ I/O's in line 6,
generation of $P$'s and $X$'s in line 7 takes $O(scan(\delta n))$ I/O's.

Procedure \texttt{genH} generates all initial $H_{P,X}$'s in $O(sort(\delta^2n))$ I/O's.
Each vertex has at most $\delta$ neighbors that succeed it in $L$, so there are $O(\delta^2n)$ tuples
added to $E$ in the first \texttt{for} loop (lines 1-5). They can be generated in
$O(scan(\delta^2n))$ I/O's, as the edges in $E$ are sorted lexicographically.

\vspace{\baselineskip}
\noindent
\fbox{
  \parbox{\linewidth}{
    \begin{algorithmic}[1]
    \REQUIRE \texttt{genH($G$, $L$)}
    \FORALL{$u$ -- vertex in $V$}
      \FORALL{$v, w$ -- neighbors of $u$ s.t.
        $u \stackrel{\scriptscriptstyle{L}}{\prec} v \stackrel{\scriptscriptstyle{L}}{\prec} w$}
        \STATE append tuple $(v, w, ?u?)$ to $E$
      \ENDFOR
    \ENDFOR
    \STATE sort $E$ lexicographically
    \FORALL{$p$ -- tuple in $E$}
      \IF{$p=(v,w,?u?)$}
        \STATE $q\leftarrow$ tuple of form $(v_i, v_j)$ immediately preceding $p$ in $E$
        \IF{$q=(v,w)$}
          \STATE output tuples $(H_v, u, w)$ and $(H_v, w, u)$ to set $H$
        \ENDIF
        \STATE do not write $p$ back to $E$
      \ELSE[$\quad//p = (v, w)$]
        \STATE write $p$ back to $E$
      \ENDIF
    \ENDFOR
    \STATE sort $H$ lexicographically
    \STATE scan $H$ to create sets $H_v$
    \STATE scan edges of $H_v$'s and mark the endpoints that belong to $P_v$'s
    \end{algorithmic}
  }
}
\vspace{\baselineskip}

To understand the meaning of tuples from line 3, refer to Fig. \ref{fig_h_graph}.

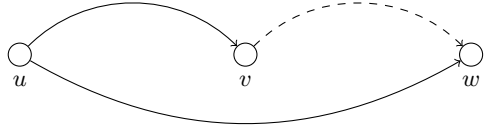
\begin{figure} 
\begin{center}
\begin{tikzpicture}

\node (u_node) at (0,0) [circle,draw] {};
\node (v_node) at (3,0) [circle,draw] {};
\node (w_node) at (6,0) [circle,draw] {};

\node [below] at (u_node.south) {$u$};
\node [below] at (v_node.south) {$v$};
\node [below] at (w_node.south) {$w$};

\draw [->] (u_node) to [bend left=45] (v_node);
\draw [->] (v_node) to [dashed,bend left=45] (w_node);
\draw [->] (u_node) to [bend right=30] (w_node);

\end{tikzpicture}
\end{center}
\caption{$v$ and $w$ are neighbors of $u$, and
$u \stackrel{\scriptscriptstyle{L}}{\prec} v \stackrel{\scriptscriptstyle{L}}{\prec} w$,
so the tuple $(v, w, ?u?)$ is output in \texttt{genH}, line 3.
If~$(v,w)\in E$, the edge $(u,w)$ is in $H_v$.}
\label{fig_h_graph}
\end{figure} 

After $E$ is sorted in line 6 in $O(sort(\delta^2n))$ I/O's, the next \texttt{for} loop (lines 7-17)
identifies edges that belong to sets $H_v$. In $E$, each edge $(v,w)$ is followed by zero or
more tuples of the form $(v,w,?u?)$. This makes it easy to determine $q$ in line 9, and therefore,
the loop takes $O(scan(\delta^2n))$ I/O's.

As explained in Fig. \ref{fig_h_graph}, $(v,w)$ followed by $(v,w,?u?)$ means that
the edge $(u,w)$ has to be added to $H_v$. It is denoted by the tuples $(H_v, u, w)$
and $(H_v, w, u)$ output in line 11.
Sorting $H$ in line 18 takes $O(sort(\delta^2n))$ I/O's,
and after that, $H_{P,X}$'s are generated
in line 19 in $O(scan(\delta^2n))$ I/O's. The marking of endpoints (line 20) also takes
$O(scan(\delta^2n))$ I/O's.

Therefore, total complexity of generating initial $P$'s, $X$'s, and $H_{P,X}$'s is
$O(sort(\delta^2n))$ I/O's. \qed

\subsection{Proof of Lemma \ref{alg_max_cliques_lemma_2}}
\vspace{\baselineskip}
\noindent
\fbox{
  \parbox{\linewidth}{
    \begin{algorithmic}[1]
    \REQUIRE \texttt{updateH($v$)}
    \FORALL{$H_{P,X}$ -- \emph{candidate}}
      \FORALL{$e$ -- edge in $H_{P,X}$}
        \IF{$e=(u,v)$ or $e=(v,w)$}
          \STATE unmark $v$ in $e$
        \ENDIF
      \ENDFOR
      \FORALL{$e$ -- edge in $H_{P,X}$}
        \IF{at least one vertex of $e$ is marked}
          \STATE write $e$ back to $H_{P,X}$
        \ELSE
          \STATE do not write $e$ back to $H_{P,X}$
        \ENDIF
      \ENDFOR
    \ENDFOR
    \end{algorithmic}
  }
}
\vspace{\baselineskip}

Recall that $|P| = p$ and $|X| = x$.
Our idea in computing $H_{P,X}$'s is to first generate \emph{candidates}
for $H_{P,X}$'s.
\emph{Candidates} are defined as $H_{P,X}$'s as they would be if there
were no lines 7 and 8 in \texttt{BronKerboschPivot} (i.e., as if $P$ and $X$ did not change).
We then update the \emph{candidates} according to lines 7 and 8 of \texttt{BronKerboschPivot}.

Generation of the \emph{candidates} for $H_{P,X}$'s is almost the same as in \texttt{genH},
only $v$ and $w$ in line 2 are now taken from $P$, so it uses $O(sort(p^2(p+x)))$ I/O's.
Updating $H_{P,X}$'s (moving $v$ from $P$ to $X$)
is realized by the procedure \texttt{updateH}.

Both inner \texttt{for} loops (lines 2-6 and 7-13) clearly take $O(scan(|H_{P,X}|))$ I/O's.
Since the total size of all \emph{candidates} is $O(p^2(p+x))$,
a single call to \texttt{updateH}
takes $O(scan(p^2(p+x)))$ I/O's. There are at most $p$ such calls, so generating
$H_{P,X}$'s that are passed to recursive calls in line 6 of \texttt{BronKerboschPivot} takes
$O(sort(p^2(p+x)) + p\cdot scan(p^2(p+x))) = O(sort(\delta p^2(p+x)))$ I/O's. \qed

\end{document}